\documentclass[a4paper,11pt]{article}
\usepackage{amsfonts}
\usepackage{amsmath}
\usepackage{amsthm}
\usepackage{amssymb}
\usepackage{epsfig}
\usepackage{mathrsfs}

\newcommand{\Z}{\mathbb{Z}}
\newcommand{\N}{\mathbb{N}}
\newtheorem{defn}{Definition}[section]
\newtheorem{thm}{Theorem}[section]

\newtheorem{rem}{Remark}[section]

\newtheorem{lem}{Lemma}[section]
\begin{document}

\title{Some Algebraic Properties Of Linear Synchronous Cellular Automata}
\author{Sreeya Ghosh and Sumita Basu}

\newcommand{\Addresses}{{
  \bigskip
  \footnotesize

Sreeya Ghosh(Corresponding author), \textsc{Dept. of Applied Mathematics; Calcutta University,
Rashbehari Shiksha Prangan,92,A.P.C. Road; Kolkata-700009.} \par \nopagebreak
Email: \texttt{sreeya135@gmail.com}

\medskip

Sumita Basu, \textsc{Bethune College;
181 Bidhan Sarani; Kolkata700006} \par \nopagebreak
 Email: \texttt{sumi\underline{ }basu05@yahoo.co.in}

}}

\maketitle

\begin{abstract}
Relation between global transition function and local transition function of a homogeneous one dimensional cellular automaton (CA) is investigated for some standard transition functions. It could be shown that left shift and right shift CA are invertible. The final result of this paper states that the set of all left and right shift CA together with the identity CA on the same set forms an abelian group.
\end{abstract}
 
\textbf{Key words : } Cellular Automaton, Transition function, Invertible Cellular Automaton.
\section{Introduction}Cellular Automata is the computational model of a dynamical system.This model was introduced by J.von Neumann and S.Ulam in 1940 for designing self replicating systems which later saw applications in Physics, Biology and Computer Science.\\
Neumann conceived a cellular automaton as a two-dimensional mesh of finite state machines called cells which are locally interconnected with each other. Each of the cells change their states synchronously depending on the states of some neighbouring cells (for details see \cite{NJV1, USM2} and references therein).  The local changes of each of the cells together induces a change of the entire mesh. Later one dimensional CA, i.e a CA where the elementary cells are distributed on a straight line was studied. Stephen Wolfram's work in the 1980s contributed to a systematic study of one-dimensional CA, providing the first qualitative classification of their behaviour ( reported in \cite{WOS3, WOS4, WOS5} ).\\ In this paper we consider only synchronous one dimensional Cellular Automaton (CA) where the underlying topology is a one dimensional grid line. A finite automaton with finite memory models a simple computation. A CA is a computation model of a dynamical system where finite/countably infinite number of automaton are arranged in an ordered linear grid. 
\begin{center}
		\begin{tabular}{c|c|c|c|c|c|c}
\hline
 ... & ... & $A_{i-1}$ & $A_i$ & $A_{i+1}$ & ... & ... \\
\hline
		\end{tabular}\\
		\vspace{0.05in}
\textit{Fig.1}\\ 
		\textit{A typical grid of a linear CA}
\end{center}
Each of the automaton work synchronously leading to evolution of the entire grid through a number of discrete time steps. If the set of memory elements of each automaton is \{0,1\} then a typical pattern evolved over time $t$ may be as follows : 
	\begin{center}
		\begin{tabular}{c|c|c|c|c|c|c|c|c|c|c}
\hline
 Time $\downarrow$ & Grid Position $(i) \rightarrow$ & ... & -3 & -2 & -1 & 0 & 1 & 2 & 3 &...\\
\hline
$t=0$ & Configuration $C^0 \rightarrow$ & ... & 0 & 0 & 0 & 1 & 0 & 0 & 0 &...\\
\hline	
 $t=1$ & $C^1 \rightarrow$ & ... & 0 & 0 & 1 & 0 & 1 & 0 & 0 &...\\
\hline	
 $t=2$ & $C^2 \rightarrow$ & ... & 0 & 1 & 0 & 1 & 0 & 1 & 0 &...\\
\hline	
 $t=3$ & $C^3 \rightarrow$ & ... & 1 & 0 & 1 & 0 & 1 & 0 & 1 &...\\
\hline
 \vdots & \vdots &... & \vdots  & \vdots  & \vdots  & \vdots  & \vdots  & \vdots  & \vdots  &...\\
\hline				
		\end{tabular}\\
		\vspace{0.05in}
\textit{Fig.2}\\
		\textit{The grid line at time $t$ gives the configuration of the CA at time $t$}
\end{center}
Algebraic properties of a CA and its relation to group theory and topology is gaining interest in recent years (see\cite{Ka05,S15,SC 04, CRG16a,CRG16b}). All the CA's referred above are deterministic in nature. Non deterministic CA have also been studied in \cite {BAS10}.\\ In this paper we study one dimensional homogeneous infinite CA and study the relationship of the macro and micro level changes of the configuration in specific cases. In Section 2 basic concepts are introduced and some fundamental results are reported. We compare the behavior of local and global transition function for some standard CA in Section 3. Section 4 is devoted to finding inverse of some standard CA. We could show that inverse of an m  place left shift function is an m {place right shift function and vice versa. Finally, in Section 5 a binary operation on the set of all CA having the same state set is defined and its properties are studied. It could be shown that the set of CA containing identity function and all shift (both left and right) functions form an abelian group under the binary function.
 
\section{Basic Concepts}
We give a formal definition of a Cellular Automaton.
\begin{defn}
Let us consider a finite set \textbf{$Q$} called the \textbf{state set}. The memory elements of the automata placed on the grid line belong to this state set $Q$.\\ A \textbf{global configuration} is a mapping from the group of integers ${\Z}$ to the set $Q$ given by $C: {\Z} \rightarrow Q$.\\
The set $Q^{\Z}$ is the set of all global configurations where $Q^{\Z} = \{C|C : {\Z} \rightarrow Q\}$.\\A mapping $\tau : Q^{\Z} \rightarrow Q^{\Z}$  is called a  \textbf{global transition function}.\\A \textbf{CA}(denoted by $\mathcal{C}_\tau^Q$) is a triplet \textbf{$(Q, Q^{\Z}, \tau)$} where $Q$ is the finite state set, $Q^{\Z}$ is the set of all configurations, $\tau$ is the global transition function.
\end{defn}
\begin{rem}For a particular state set $Q$ and a particular global transition function $\tau$  a triple \textbf{($Q, Q^{\Z}, \tau$)} denoted by $\mathcal{C_\tau^Q}$ defines the set of all possible cellular automata on  ($Q,\tau$). However, the evolution of a CA at times is dependent on the initial configuration (starting configuration) of the CA. A particular CA $\mathcal{C_\tau^Q}(C^0)\in\mathcal{C_\tau^Q}$ is defined as the quadruple \textbf{($Q, Q^{\Z}, \tau, C^0$)} such that $C^0\in Q^{\Z}$ is the initial configuration of the particular CA $\mathcal{C_\tau^Q}(C^0)$. 
\end{rem}
The configuration at time $t$ is denoted by $C^t  $ such that $C^t\in Q^{\Z}$ for all time $t$ .\\Also,$$\tau(C^t) = C^{t+1}$$
With reference to Fig 2 we get,\\ $C^0 = ...0001000...$;~~$\tau(C^0) = \tau(...0001000...) = ...0010100... = C^1 $;\\
$\tau(C^1) = \tau(...0010100...) = ...0101010... = C^2$,etc.\\\\
Evolution of a Cellular Automaton is mathematically expressed by the global transition function. However, this global transition is induced by transitions of automaton at each grid point of the CA. The transition of the state of the automaton at the ith grid point of a CA at a particular time, depends on the state of the automaton at the $i^{th}$ grid point and its adjacent cells. These adjacent cells constitute the neighbourhood of that cell. The transition of the automaton at each grid point is called local transition.
\begin{defn}
For $i\in {\Z}, r\in {\N}$,  let $S_i = \{i-r,...,i-1,i,i+1,...,i+r \}\subseteq {\Z}$. $S_i$ is the neighbourhood of the $i^{th}$ cell. r is the radius of the neighbourhood of a cell. It follows that ${\Z} = \bigcup_iS_i$\\A restriction from ${\Z}$ to $S_i$ induces the following:
\begin{enumerate}
\item Restriction of $C$ to $c_i$ is given by $c_i: S_i \rightarrow Q$; and $c_i$ may be called \textbf{local configuration} of the $i^{th}$ cell.\\
\item Restriction of $Q^{\Z}$ to $Q^{S_i}$ is given by $Q^{S_i}=\{c_i|c_i : {S_i} \rightarrow Q\}$;and $Q^{S_i}$ may be called the \textbf{set of all local configurations} of the $i^{th}$ cell.
\end{enumerate}
The mapping $\mu_i : Q^{S_i} \rightarrow Q$ is known as a \textbf{local transition function} for the $i^{th}$ automaton having radius $r$. So, $ \forall i \in { \Z},\mu_i (c_i)\in Q $. So,if the local configuration of the ith cell at time $t$ is denoted by $c^{t}_i$ then $\mu_i(c^{t}_i)=c^{(t+1)}_i(i)$.\\
\end{defn}
\begin{rem}Since $S_i$ has $(2r+1)$ elements and $Q$ is finite, $Q^{S_i}$ is finite having $|Q|^{(2r+1)}$ elements.
\end{rem}
\begin{rem}From Definition 2.2 ,
\begin{enumerate}
\item  $\forall i\in {\Z}$, $C(i) = c_i(i)$\\ 
\item $\forall j \in S_i \subseteq {\Z}, c_i(j) = C(j).$\\
\end{enumerate}
\end{rem}
\begin{rem}If $\tau(C) = C^*$ then $ C^*(i) = \tau(C)(i) = \mu_i(c_i)$. So we have, 
\begin{enumerate}
\item $C^{(t+1)}(i)=\tau(C^t)(i)=\mu_i(c^{t}_i)=c^{(t+1)}_i(i)$\\
\item$ \tau(C)=.....\mu_{i-1}(c_{i-1}).\mu_i(c_i).\mu_{i+1}(c_{i+1}).......$
\end{enumerate}
\end{rem}
\begin {defn}If all $\mu_i's$  are identical then the CA is \textit{homogeneous}.A \textbf{homogeneous Cellular Automaton} may also be defined as a triplet \textbf{$(Q, r, \mu)$} where $Q$ is the finite state set, $r$ is the radius of the neighbourhood of a cell, $\mu$ is the local transition function.
\end{defn}
Henceforth, in this paper by a CA  we will consider one dimensional homogeneous synchrounous CA.
\begin{rem}The set $\mathscr{Q}^{\Z}$, where $\mathscr{Q}^{\Z}= (Q^{\Z})^{Q^{\Z}}$ is the set of all global transition functions of a \textit{CA} defined as $\mathscr{Q}^{\Z}= (Q^{\Z})^{Q^{\Z}} = \{\tau|\tau : Q^{\Z} \rightarrow Q^{\Z}\}$.
\end{rem}
\begin{rem}The set M, where $M =\{ (Q^{S_i})^{Q}|i \in {\Z}\}$ is the set of all local transition functions of a \textit{CA} defined as $M=\{(Q^{S_i})^{Q}|i \in {\Z}\} = \{\mu|\mu : Q^{S_i} \rightarrow Q, i \in {\Z}\}$.
\end{rem}
If there is no ambiguity regarding $Q, C^0$ a CA is often denoted by $\tau$ where $\tau \in \mathscr{Q}^{\Z}$. 
\begin{defn} If for a particular CA, $ |Q|=2$ so that we can write $Q=\{0, 1\}$, then the CA is said to be a binary CA.\\For a binary CA $(Q, Q^{\Z}, \tau)$ if $C_1, C_2 \in Q^{\Z}$ such that $\tau(C_1) = C_2$ where $$ \forall i \in {\Z},~C_1(i)=0 \leftrightarrow C_2(i) =1~and ~C_1(i)=1 \leftrightarrow C_2(i) =0$$ then $C_1$ is the complement of $C_2$ and vice versa. $\tau$ is said to be the complementary transition function and is denoted by $\tau^c$.
\end{defn}
\begin{defn}A local transition function $\mu$ is an \textbf{identity function} denoted by $\mu_{e}$ if,  $\forall  i \in {\Z}, ~ \mu(c_i) =  c_i(i)$ \\
A global transition function $\tau_e$ is an \textbf{identity function} provided for all $C \in Q^{\Z}, \tau_e(C) = C$. \\
A CA is said to be an \textbf{identity Cellular Automaton} if the global transition function is an identity function.
\end{defn}
\begin{defn}A local transition function $\mu$ is a \textbf{constant function} denoted by $\mu_{q}$ if, $\forall  i \in {\Z}, ~ \mu(c_i) =  q$  for a particular state $q \in Q$.\\
A global transition function $\tau$ is a \textbf{constant function} provided for all $C \in Q^Z, \tau(C) = C^*$ for a particular constant configuration $C^* \in Q^{\Z}$.\\A CA is said to be a \textbf{constant Cellular Automaton} if the global transition function is a constant function.
\end{defn}
\begin{defn}A local transition function $\mu$ is an \textbf{m-place left shift function} denoted by $\mu_{Lm}$ where $m \in {\N}$ is finite, if the state of the $i^{th}$ automaton $c_i(i)$ shifts $m-place$ leftwards. So, $ \forall i \in {\Z}, \mu_{Lm}(c_i) = c_i(i+m)$ where $c_i$ is the restriction of $C$.\\
A global transition function $\tau$ is an \textbf{m-place left shift function} denoted by $\tau_{Lm}$ where $m \in {\N}$ is finite, if $ \forall i \in {\Z}, \tau_{Lm}(C)(i) = C(i+m)$.\\A CA is said to be a \textbf{m-place left shift Cellular Automaton} if the global transition function is a m-place left shift function.

\end{defn}
\begin{defn}A local transition function $\mu$ is an \textbf{m-place right shift function} denoted by $\mu_{Rm}$ where $m \in {\N}$ is finite, if the state of the $i^{th}$ automaton $c_i(i)$ shifts $m-place$ rightwards. So, $ \forall i \in {\Z}, \mu_{Rm}(c_i) = c_i(i-m)$ where $c_i$ is the restriction of $C$.\\
A global transition function $\tau$ is an \textbf{m-place right shift function} denoted by $\tau_{Rm}$ where $m \in {\N}$ is finite, if $ \forall i \in {\Z}, \tau_{Rm}(C)(i) = C(i-m)$.\\A CA is said to be a \textbf{m-place right shift Cellular Automaton} if the global transition function is a m-place right shift function.
\end{defn}
\begin{thm}
For a homogeneous CA, if the global transition function $\tau$ is a constant function then $\forall i \in {\Z}$, $\tau(C)(i)$ is identical. 
\end{thm}
\begin{proof}
Let the global transition function $\tau$ of a homogneous CA be a constant function.\\
Then for all $C \in Q^{\Z}$, $$\tau(C) = \tau_q(C) = C^*$$ where $C^* \in Q^{\Z}$ is a particular configuration.\\
Thus if  $\forall i \in {\Z}$, $\mu=\mu_i$ be the corresponding  local transition function then $$\tau(C)(i) = C^*(i) = \mu_i(c_i) = q_i \in Q$$ 
Now, for $j,k \in {\Z}$,$$ \tau(C)(j)  \neq \tau(C)(k) \implies C^*_j \neq C^*_k  \implies \mu_j(c_j) \neq \mu_k(c_k) \implies \mu(c_j) \neq \mu(c_k)$$
This is a contradiction to the fact that the CA is homogeneous.\\
Hence the theorem.
\end{proof}
\section{Relation between Local and Global Transition Function of a Homogeneous Cellular Automaton}
The global transition function $\tau$ describes the evolution of the dynamical system at the macro level. The corresponding local transition function($\mu$) having radius of the neighbourhood $r$ denoted by $\mu^ r$ describes the same at micro level.
\begin{thm}For a homogeneous CA, having the global transition function $\tau$ the corresponding $\mu^r$ may not be unique.
\end{thm}
\begin{proof}The result is proved by example.\\Let,  $ x_n = 2^{3n}.$ Then the squence$\{x_n\}$ in decimal system is $\{1,8, 64.....\}$\\The binary representation of this sequence is $\{1,1000, 1000000, 1000000000, .....\}$. This sequence can be generated by the CA with state set $Q=\{0,1\}~and~ \tau= \tau_{L3}$. Let $\mu^r : Q^{S_i} \rightarrow Q$ be the local transition function of the CA having radius $r \in {\N}$.\\ Depending on the value of r we have thefollowing cases.
\begin{itemize}
	\item Case 1 : $r<3$ say $r=2$\\
	Since $c_i(i+3)$ is not defined, $\mu^2(c_i) \neq c_i(i+3)$.	So, the local transition cannot be represented by $\mu^2$. Similarly $\mu^1$ cannot be the local transition function.
	\item Case 2 : $r \geq 3$\\ 
In this case $c_i(i+3)$ is defined, and $C(i+3)=c_i(i+3)$.  
	So, $\forall r \geq 3$, we may define, $ \mu^r(c_i)=c_i(i+3)$ and $\mu^r$ may act  as local transition function of the CA\\
	\end{itemize}
Case 2 shows that for a given global transition function the local transition function is not unique.
\end{proof}
\begin{lem} Given a local transition function $\mu^ r$ the corresponding global transition function $\tau$ is  unique. 
\end{lem}
\begin{thm}
For a homogeneous CA, the global transition function of the CA is an identity function $\tau_e$ if and only if the local transition function is an identity function $\mu_{e}$. 
\end{thm}
\begin{proof}
Let $\mu$ be the local transition function and $ \tau$be the global transition function of a homogeneous CA.\\ 
Suppose $\mu$ is an identity function. Then for any $i \in {\Z}$, we have, $\mu(c_i) = \mu_e(c_i) = c_i(i)$\\
If  $\tau(C) = C^*$ then $\forall i \in {\Z}$ we have, $$\tau(C)(i) = C^*(i) =  \mu(c_i)= c_i(i)$$
Again,  $ \forall i \in  {\Z}$ we know that, $c_i(i) = C(i).$ So, $\forall i \in  {\Z},$\\
$$C^*(i) = c_i(i) = C(i) \Leftrightarrow  \tau(C)(i) = C(i)$$
Since the CA is homogeneous, we have, $\tau(C) = C$.\\
It follows that the global transition function is an identity function and is denoted by $\tau_e$.\\
Conversely, let the global transition function be an identity function.\\
Then, $\tau(C) = \tau_e(C) = C $ Since, $i \in  {\Z}$,\\ 
 $$C(i) = c_i(i) \Leftrightarrow  \tau_e(C)(i) = \mu(c_i) = c_i(i)$$
Therefore the local transition function is an identity function.\\
Hence the theorem. 
\end{proof} 

\begin{thm}
For a homogeneous CA, the global transition function of the CA is a constant function $\tau_q$ if and only if the local transition function is a constant function $\mu_{q}$. 
\end{thm}
\begin{proof}Let $\mu$ be the local transition function and $ \tau$be the global transition function of a homogeneous CA.\\ Let the local transition function $\mu$ be a constant function $\mu_{q}$. Then for any $i \in {\Z}$ there exists a partcular $q \in Q$ such that\\
$$\mu(c_i) = \mu_q(c_i) = q$$
Now, if $\tau(C) = C^*$ then $\forall i \in {\Z}$ we have, $$\tau(C)(i) = C^*(i) =  \mu(c_i)=q$$
So $\forall i \in {\Z},~~ C^*(i) = q $. Thus $C^*$ is a constant configuration.\\
Therefore, for the constant configuration $C^* \in Q^{\Z}$, and forall $C \in Q^{\Z}, \tau(C) = C^*$ .\\
It follows that the global transition function is a constant function.\\
 Conversely, let the global transition function be a constant function $\tau_q$.\\
Then for all $C \in Q^{\Z},~ \tau(C) = \tau_q(C) = C^*$ for a particular configuration $C^* \in Q^{\Z}$ such that $\forall i \in {\Z}~~C^*(i)=q.$\\
Again, $\forall i \in {\Z},~~\tau(C)(i) = \mu(c_i) =q $.\\
Therefore it follows that the local transition function is a constant function and is denoted by $\mu_q$.\\
Hence the theorem.
\end{proof}
\begin{thm}
For a homogeneous CA, the global transition function of the CA is an $m-place$ left shift function $\tau_{Lm}$ if and only if the local transition function is an $m-place$ left shift function $\mu_{Lm}$ where $m \in {\N}$ is finite. 
\end{thm}
\begin{proof}Let $\mu^r$ be the local transition function and $ \tau$be the global transition function of a homogeneous CA. \\Let the local transition function be an $m-place$ left shift function where $m \leq r$.Then for any $i \in {\Z}$,if $\tau(C) = C^*$\\
$$\mu^r(c_i) =( \mu^r)_{Lm}(c_i) = c_i(i+m)=C(i+m)$$
$$\tau(C)(i) = C^*(i) =  \mu^r(c_i)= C(i+m)$$
Hence it follows that the global transition function is an $m-place$ left shift function where $m \leq r$ such that $r \in {\N}$ is finite.\\
Conversely, let the global transition function be an $m-place$ left shift function.\\
Then, $$\forall C \in Q^{\Z}, ~\forall i \in {\Z},~~\tau(C)(i) = \tau_{Lm}(C)(i) = C(i+m)  $$
However,$$\forall i \in {\Z},m \leq r ~ \mu^r(c_i)=\tau(C)(i) =  C(i+m)=c_i(i+m)$$
Thus it follows that  $\mu$ the local transition function of the CA  is an $m-place$ left shift function where $m \in {\N}$ is finite.\\
Hence the theorem.
\end{proof}
Similarly we have the following result which we state without proof.
\begin{thm}
For a homogeneous CA, the global transition function of the CA is an $m-place$ right shift function $\tau_{Rm}$ if and only if the local transition function is an $m-place$ right shift function $\mu_{Rm}$ where $m \in {\N}$ is finite. 
\end{thm}
\section{Inverse of a Homogeneous Cellular Automaton}
\begin{defn} Let a class of CA be given by $(Q, Q^{\Z}, \tau)$. Any global transition function $\tau^{-1} \in (Q^{\Z})^{Q^{\Z}} $such that for all $C_{i}, C_{j} \in Q^Z$,\\ $$\tau(C_{i}) =C_{j} \Leftrightarrow \tau^{-1}(C_{j}) =C_{i}$$ is called the inverse of the global transition function $\tau$. \\Consequently, a CA given by $(Q, Q^{\Z}, \tau^{-1})$ is the inverse of the CA $(Q, Q^{\Z}, \tau)$.
\end{defn}
\begin{thm}
The inverse of an identity function is an identity function.    
\end{thm}
\begin{proof}
Let us consider an identity global transition function $\tau_{e}$.\\
Therefore , $\forall C \in Q^{\Z}$,$$\tau_e(C) = C  $$
If $\tau_e^{~-1}$ denotes the inverse the of $\tau_e$, then we have,$$\tau_e(C) =C \Leftrightarrow \tau_e^{-1}(C) =C$$
Thus, for all $C \in Q^{\Z}$, $$\tau_e(C) = \tau_e^{-1}(C) \Leftrightarrow \tau_e =\tau_e^{-1}$$
Hence the theorem.
\end{proof} 
\begin{thm} For a binary homogeneous CA the inverse of the complementary transition function is the function itself.
\end{thm}
\begin{proof}By definition, $$\tau^c(C_1) = C_2 \leftrightarrow \tau^c(C_2)= C_1$$ Hence $(\tau^c)^{-1} =\tau^c$
\end{proof}
\begin{thm}
The inverse of an $m-place$ left shift global transition function $\tau_{Lm}$ where $m \in {\N}$ is finite, is an $m-place$ right shift global transition function $\tau_{Rm}$ and vice-versa.    
\end{thm}
\begin{proof}
Let us consider an \textit{m-place} left shift global transition function $\tau_{Lm}$.\\ 
Then,$$\forall i \in {\Z},~ \tau_{Lm}(C)(i) = C(i+m)$$
So, $$\tau_{Lm}(C)(i-m) = C(i-m+m) = C(i)$$
Again, for an \textit{m-place} right shift global transition function $\tau_{Rm}$,we have,
$$\forall i \in {\Z},~ \tau_{Rm}(C)(i) = C(i-m) $$
So, $$\tau_{Rm}(C)(i+m) = C(i+m-m) = C(i)$$
Thus ,$\forall i \in {\Z}$, $$\tau_{Lm}(C)(i) = C(i+m) \Leftrightarrow \tau_{Rm}(C)(i+m) = C(i)$$
Therefore, we conclude that the inverse of $\tau_{Lm}$ is $\tau_{Rm}$.\\
Hence, the theorem. 
\end{proof}	
\section{Binary Operations on One Dimensional Cellular Automata with the Same State Set}
In this section we discuss the properties of binary operations on the set of CA with the same state set $Q$. Let, $\mathscr{Q}^{\Z}= (Q^{\Z})^{Q^{\Z}}$. Thus any element of the set $\mathscr{Q}^{\Z}$ is a global transition function or often called a CA with the state set $Q$.
\begin {defn} A binary operation $*  $ on $\mathscr{Q} ^{\Z}$ is defined as follows:\\ If $\tau_1 ~and ~\tau_2 \in \mathscr{Q}^{\Z}$, then $$\forall C \in Q^{\Z}, (\tau_1 * \tau_2 )(C) = \tau_1(\tau_2(C))$$
\end{defn}
\begin{lem} For a particular state set $Q,~ \tau_e$, the global identity transition function belongs to $\mathscr{Q}^{\Z}$.
\end{lem}
\begin{lem} For a particular state set $Q,~\forall m,n  \in \N,~\tau_{Lm}, \tau_{Rn}$, the global m-place left shift transition function and global n-place right shift function belongs to $\mathscr{Q}^{\Z}$.
\end{lem}
\begin{lem}For a particular state set $Q,~\mathscr{Q}^{\Z}$ is closed under the binary operation $*$.
\end{lem}
\begin{proof}As $\tau_1 ~and ~\tau_2 \in \mathscr{Q}^{\Z}$, $\forall C_i \in Q^{\Z}, \exists C_j\in Q^{\Z}$ so that $\tau_2(C_i) = C_j$. Hence,
$$ (\tau_1 * \tau_2 )(C_i) = \tau_1(\tau_2(C_i))= \tau_1(C_j) \in Q^{\Z}$$
So, $ (\tau_1 * \tau_2 ) \in \mathscr{Q}^{\Z}$
\end{proof}
\begin{lem} For a particular state set $Q,~ *$, the binary operation on $\mathscr{Q}^{\Z}$ is associative.
\end{lem}
\begin{proof}Let $\tau_1 , \tau_2 ~and ~\tau_3 \in \mathscr{Q}^{\Z}$, then $ \forall C \in Q^{\Z}$,$$\tau_1* ((\tau_2 * \tau_3 )(C) )= \tau_1*((\tau_2(\tau_3(C)))= \tau_1((\tau_2(\tau_3(C))$$
$$((\tau_1* \tau_2) * \tau_3 )(C)  = ( \tau_1*\tau_2)\tau_3(C)= \tau_1((\tau_2(\tau_3(C))$$
\end{proof}
\begin{lem}For a binary CA if $\tau^c$ is the complementary transition function then $\ \tau^c \in \mathscr{Q}^{\Z}$.
\end{lem}
\begin{lem} For a particular state set $Q,~\forall m,n  \in \N,~\tau_{Lm}, \tau_{Rn}$, the global m-place left shift transition function and global n-place right shift function then 
\begin{enumerate}
\item $\tau_{Lm}*\tau_{Rn}= \tau_{Rn}*\tau_{Lm}$.\\
\item If, $m>n, ~\tau_{Lm}*\tau_{Rn}$ is a left  shift function.\\
\item If $n>m, ~\tau_{Lm}*\tau_{Rn}$ is a right shift function.
\end{enumerate}
\end{lem}
\begin{proof}$\forall C \in Q^{\Z}, $ we have
$$(\tau_{Lm}*\tau_{Rn})C(i)= \tau_{Lm}(C(i-n))=C(i-n+m)~and~ (\tau_{Rn}*\tau_{Lm})C(i)=\tau_{Rn}(C(i+m))=C(i+m-n)$$Hence we get result 1. If $m-n =p, ~\tau_{Lm}*\tau_{Rn}$ is a left  shift function when $p> 0$ and right shift when $p<0$.
\end{proof}
Lemma 5.1, 5.3 and 5.4 gives the following theorem.
 \begin{thm}$ \langle \mathscr{Q}^{\Z}, * \rangle $  forms a monoid.
\end{thm}
From Lemma 5.5, Theorem 4.2 and Lemma 5.3 we have the following,
\begin{thm}For a binary CA if $\tau^c$ is the complementary transition function then $ \langle \{\tau_e,\tau^c \}, * \rangle $ forms a group and is the smallest such nontrivial group.
\end{thm}
Lemma 5.2, Lemma 5.6, Theorem 4.3 and Theorem 5.1 together gives the final result given below.
\begin{thm}
Let us consider the set $G=\{\tau_e\} \bigcup \{\tau_{Lm}|m \in {\N}\} \bigcup \{\tau_{Rm}|m \in {\N}\}$ where $\tau_{Lm}$ is a m-place left-shift function and $\tau_{Rm}$ is a m-place right-shift function. Then $(G, *)$ forms an abelian group with respect to binary operation $*$.
\end{thm}
\section{Conclusion} The results obtained in this paper is for homogeneous and linear Cellular Automaton. An investigation /extension of results for other types of Cellular Automaton may be worth attempting.

\Addresses

\end{document}